\documentclass[preprint]{elsarticle}

\usepackage{etex}   

\usepackage{lineno}
\modulolinenumbers[5]
\usepackage[T1]{fontenc}
\usepackage{times}%
\usepackage{microtype}
\usepackage{color}
\usepackage{amsfonts}
\usepackage{amsthm}
\usepackage{amssymb}
\usepackage{pictexwd}
\usepackage{fancyhdr}

\usepackage{mathrsfs}
\usepackage{enumitem}
\usepackage{natbib}

\usepackage{marginnote}

\newtheorem{teo}{Theorem}

\newtheorem{cor}{Corollary}
\newtheorem{prop}{Proposition}
\newtheorem{defi}{Definition}
\newtheorem{exem}{Example}
\newtheorem{rmk}{Remark}

\begin{document}
	
	\begin{frontmatter}
		
		\title{Bounded Generalized Mixture Functions\footnote{Preprint submited to Information Fusion}}
		
		\author[ufrn,ufersa]{Antonio Diego S. Farias\corref{cor2}}
		\ead{antonio.diego@ufersa.edu.br}
		
		\author[ufrn]{Valdigleis S. Costa\corref{cor1}}
		\ead{valdigleis@gmail.com}
		
		\author[ufrn]{Luiz Ranyer de Ara\'ujo Lopes}
		\ead{ranyer.lopes@gmail.com}
		
		\author[ufrn]{Benjam\'in Bedregal}
		\ead{bedregal@dimap.ufrn.br}
		
		\author[ufrn]{Regivan H. N. Santiago}
		\ead{regivan@dimap.ufrn.br}
				
		\cortext[cor1]{Corresponding author}
		
		\cortext[cor2]{Principal corresponding author}
		
		\address[ufrn]{Group for Logic, Language, Information, Theory and Applications \\Federal University of Rio Grande do Norte - UFRN, Department of Informatics and Applied Mathematics, Avenue Senador Salgado Filho, 3000, Universitary Campus of Lagoa Nova, 59.078-970, Natal, RN, Brazil}
		
		\address[ufersa]{Group of  Theory of Computation, Logic and Fuzzy Mathematics \\Federal Rural University of Semi-Arid - UFERSA, Department of Exacts and Naturals Sciences\\ Highway 226, KM 405,  S\~ao Geraldo, 59.900-000, Pau dos Ferros, RN, Brazil.\footnote{Universidade Federal Rural do Semi-\'Arido - UFERSA, Campus Pau dos Ferros, BR 226, KM 405, S\~ao Geraldo, 59.900-000, Pau dos Ferros, RN, Brasil.}}
		
		\begin{abstract}
			In literature, it is common to find problems which require a way to encode a finite set of information into a single data; usually means are used for that. An important generalization of means are the so called  Aggregation Functions, with a noteworthy subclass called  \textsf{OWA} functions. There are, however, further functions which are able to provide such codification which do not satisfy the definition of aggregation functions; this is the case of pre-aggregation and mixture functions. 
			
			In this paper we investigate two special types of functions:  \textit{Generalized Mixture}  and \textit{Bounded Generalized Mixture} functions. They generalize both: \textsf{OWA} and Mixture functions. Both Generalized and Bounded Generalized Mixture functions are developed in such way that the weight vectors are variables depending on the input vector. A special generalized mixture operator, $\mathbf{H}$, is provided and applied in a simple toy example.			
		\end{abstract}
		
		\begin{keyword}
			\texttt{Aggregation functions}\sep \texttt{Pre-aggregation functions}\sep \texttt{OWA functions}\sep \texttt{Mixture functions} \sep \texttt{Generalized Mixture functions} \sep \texttt{Bounded Generalized Mixture functions}
		\end{keyword}
		
	\end{frontmatter}

\section{Introduction}

Several applications require the conversion of a finite  collection of data (of same type) into a single data  \cite{Dubois2004,FariasJIFS2016,Paternain2012,Paternain2015}. Some abstract tools which deal with that are the so called \textit{Aggregation Functions} and \textit{Mixture Functions} \cite{Gleb2016}. Yager \cite{Yager1988} introduced a special class of aggregation functions, called  \textit{Orde Weighted Averaging} - \textsf{OWA}, and ever since several authors have proposed generalizations for them. Two generalizations are:  (1) \textit{Mixture functions}  \cite{Gleb2016} and (2) \textit{Generalized Mixture functions} \cite{Pereira1999,Pereira2000,Pereira2003}. They are not always aggregation functions, since some of them do not satisfy the monotonicity, however they are  also efficient in order to ``codify'' a set of information into a sigle one. Some other extensions were also proposed and can be found in: \cite{Yager2006,LIZASOAIN2013,LLAMAZARES2015131,XU2006231,YAGER2010374}. 

In this paper we investigate \textit{Generalized Mixture functions} (\textsf{GM}), which are weighted averaging means whose weights are dynamic; namely the weights are not fixed beforehand but depend on the input variables. This provides a  more flexible usage of weights which is not possible for functions like OWAs. The resulting functions do not require the property of monotonicity, essencial in aggregations, instead they require diretional monotonicity\cite{Bustince2015}. Further we investigate the weakening of condition  $\sum\limits_{i=1}^n w_i=1 \mbox{ to } \sum\limits_{i=1}^n w_i\le 1$, thereby obtaining another generalization of \textsf{OWAs}, called {\it Bounded Generalized Mixture} - \textsf{BGM} functions. This paper ends with the proposal of a special \textsf{GM} function,  $\mathbf{H}$, which have a wide range of properties like: idempotency, symmetry, homogeneity and diretional monotonicity. It is applyied on a simple toy example. 

This work is structured in the following way: The next section provides the basic concepts of Aggregation functions; section 3 introduces the concepts of \textit{Generalized Mixture} (\textsf{GM}) and \textit{Bounded Generalized Mixture}  (\textsf{BGM}) operators, it shows properties, constructions, examples and propose a particular \textsf{GM} function: $\mathbf{H}$; section 4 presents the final remarks and section 5 an illustrative application for \textsf{GM}.

\section{Aggregation Functions}

Aggregation functions are important mathematical tools for applications in various fields, such as: Information fuzzy \cite{FariasJIFS2016,Hancer2015,Lingling2015}; Decision making \cite{Paternain2012,Bustince2013,Yager2011}; Image processing \cite{Paternain2015,Beliakov2012,Joseph2014} and Engineering \cite{Liang2014}.

\begin{defi}
	An $n$-ary aggregation function is a mapping $A:[0,1]^n\rightarrow[0,1]$, which associates each $n$-dimensional vector ${\bf x}=(x_1,\ldots,x_n)$ to a single value $A({\bf x})$ and satisfies conditions of: (A1) mononicity (A2) Boundary:
	\begin{enumerate}
		\item[(A1)] If ${\bf x}\le {\bf y}$, i.e., $x_i\le y_i$, for all $i=1,2,...,n$, then $A({\bf x})\le A({\bf y})$;
		\item[(A2)] $A(0,...,0)=0$ and $A(1,...,1)=1$.
	\end{enumerate}
\end{defi}

\begin{exem}
	Given ${\bf x}=(x_{1}, \ldots, x_{n})$,
	\begin{enumerate}
		\item[(a)] Arithmetic Mean: $Arith({\bf x}) = \displaystyle\frac{1}{n}(x_{1} + x_{2} ... + x_{n})$
		
		\item[(b)] Minimum: $Min({\bf x}) = min\{ x_{1}, x_{2},..., x_{n} \}$;
		
		\item[(c)] Maximum: $Max({\bf x}) = max\{ x_{1}, x_{2},..., x_{n} \}$;
		
		\item[(d)] Product: $Prod({\bf x}) =\prod\limits_{i=1}^{n}x_i$;
		
		\item[(e)] Weighted Average: For ${\bf w}=(w_1,\cdots,w_n)\in[0,1]^n$, with $\sum\limits_{i=1}^n w_i=1$,\linebreak $WAvg_{\bf w}({\bf x})=\sum\limits_{i=1}^nw_i\cdot x_i$.
	\end{enumerate}
\end{exem}

\begin{rmk}
	From now on we will use the short term ``aggregation"  instead of ``$n$-ary aggregation function". 
\end{rmk}

Aggregations can be divided into four distinct classes: \textit{Averaging, Conjunctive, Disjunctive} and \textit{Mixed}. A wider description about them can be found in \cite{Beliakov2016,Dubois2000,Grabisch2009}. In this work, we only study averaging functions which satisfy the following property:

\begin{defi}
	A function $f:[0,1]^n\longrightarrow [0,1]$ satisfies the {\bf averaging property}, if for all ${\bf x}\in [0,1]^n$:
	$$Min({\bf x}) \leq f({\bf x}) \leq Max({\bf x}).$$
\end{defi}

When aggregations satisfy the averaging property we say that they are {\bf averaging aggregations}. The functions $Min$, $Max$, $Arith$ and $WAvg$ are examples of averaging aggregations. Besides that, an aggregation $A:[0,1]^n\rightarrow[0,1]$ can satisfy:

\begin{enumerate}	
	\item[(1)] {\bf Idempotency}, i.e., $A(x,..., x) = x$ for all $x \in [0,1]$;
	
	\item[(2)] {\bf Homogeneity}  of order $k$ , i.e., for all $\lambda \in [0,1]$ and ${\bf x} \in [0,1]^{n}$,\linebreak $f(\lambda x_{1},\lambda x_{2},..., \lambda x_{n}) = \lambda^k f(x_{1}, x_{2},..., x_{n})$. When $A$ is homogeneous of order $1$ we simply say that $f$ is homogeneous;
	
	\item[(3)] {\bf Shift-invariance}, i.e., $f(x_{1} + r, x_{2} + r,.., x_{n} + r) = f(x_{1}, x_{2},.., x_{n}) + r$, for all $r\in [-1,1]$, ${\bf x} \in [0, 1]^{n}$,  $(x_1+r,x_2+r,...,x_n+r)\in [0,1]^n$ and $f(x_1,x_2,...,x_n)+r\in [0,1]$;
	
	\item[(4)] {\bf Monotonicity}, i.e., $f(x_1,\cdots,x_n)\le f(y_1,\cdots,y_n)$ whenever $x_i\le y_i$, for all $i\in\{1,\cdots,n\}$;
	
	\item[(5)] {\bf Strict monotonicity}, i.e., $f({\bf x})<f({\bf y})$ whenever ${\bf x}<{\bf y}$, i.e., ${\bf x}\le {\bf y}$ and ${\bf x}\ne {\bf y}$, where ${\bf x}=(x_1,\cdots,x_n)$ and ${\bf y}=(y_1,\cdots,y_n)$;
	
	\item[(6)] {\bf Symmetry}, i.e., its value is not changed under the permutations of the coordinates of ${\bf x}$, i.e.,
	$$f(x_{1},..., x_{n}) = f(x_{\sigma_{(1)}}, \cdots, x_{\sigma_{(n)}})$$
	for all $x$ and any permutation $\sigma:\{1,...,n\}\rightarrow\{1,..., n\}$;
	
	\item[(7)] The existence of {\bf neutral element}, i.e., there is $e \in [0,1]$, such that for $t \in [0,1]$ at any coordinate  of the input vector ${\bf x}$, it has to be:
	$$f(e,..., e, t, e,...,e) = t;$$
	
	\item[(8)] The existence of {\bf absorbing element} or \textbf{(\textit{annihilator})}, i.e., there is $a \in [0,1]$, such that
	$$f(x_{1},..., x_{i-1}, a, x_{i+1}, ..., x_{n}) = a;$$
	
	\item[(9)] The existence of {\bf zero divisor}, i.e., there is $a\in\ ]0,1[$, such that for any ${\bf x}\in]0,1]^n$, with $ a $ in one of its coordinates it is verified that $f({\bf x}) = 0$;
	
	\item[(10)] The existence of  {\bf one divisor}, i.e., there is $a \in [0, 1[$, such that for any ${\bf x}\in]0,1]^n$, with $a$ in one of its coordinates it is verified that $f({\bf x}) = 1$.
\end{enumerate}

\begin{exem}\hfill
	\begin{enumerate}
		\item[(i)] The functions: $Arith, Min$ and $Max$ are examples of idempotent, homogeneous, shift-invariant and symmetric aggregations.
		\item[(ii)] $Min$ and $Max$ have the elements $0$ and $1$ as its respective annihilators, whereas $Arith$ does not have annihiladors.
		\item[(iii)] $Min$, $Max$ and $Arith$ does not have zero divisors and one divisors.
	\end{enumerate}	
\end{exem}

\subsection{Ordered Weighted Averaging - \textsf{OWA} Functions}

In the field of aggregations there is a very important kind of function in which the aggregation behavior is provided parametrically; they are called: \textit{\textbf{Ordered Weighted Averaging}} or simply \textsf{OWA} \cite{Yager1988}.

%More precisely, they are average aggregation whose behavior is in function of a vector of weights. Observe the following definition.

\begin{defi}
	Let be an input vector ${\bf x}=(x_1,x_2,\dots,x_n)\in [0,1]^n$ and a vector of weights ${\bf w}=(w_1,\dots, w_n)\in[0,1]^n$, such that $\sum\limits_{i=1}^n w_i=1$. Assuming a permutation of $\bf x$: 
	$$sort({\bf x})=(x_{(1)},x_{(2)},\dots, x_{(n)})$$
	such that $x_{(1)}\geq x_{(2)}\geq\dots\geq x_{(n)}$. The Ordered Weighted Averaging (\textsf{OWA}) function with respect to ${\bf w}$, is the function $\textsf{OWA}_{\bf w}:[0,1]^n\rightarrow [0,1]$ defined by:
	$$\textsf{OWA}_{\bf w}({\bf x})=\sum\limits_{i=1}^n w_i\cdot x_{(i)}$$
\end{defi} 

\begin{rmk}
	In what follows we remove ${\bf w}$ from $\textsf{OWA}_{\bf w}({\bf x})$ and write only \textsf{OWA}.
\end{rmk}

Examples and properties of \textsf{OWA} functions can be found in \cite{Gleb2016}. Here is important to note that $Min$, $Max$, $Arith$ and {\it Median} (described bellow) are examples of \textsf{OWA}.

\begin{exem}
	Given a vector ${\bf x}$ and its ordered permutation $sort({\bf x}) = (x_{(1)},\ldots,x_{(n)})$, the {\it Median} function
	$$Med({\bf x})= \left\{\begin{array}{ll}
	\frac{1}{2}(x_{(k)}+x_{(k+1)}), & \mbox{ if } n=2k\\
	x_{(k+1)}, & \mbox{ if } n=2k+1
	\end{array}\right.$$
	is an $\textsf{OWA}$ function in which the vector of weights is defined by:
	\begin{itemize}
		\item If $n$ is odd, then $w_i=0$ for all $i\ne\frac{n+1}{2}$ and $w_{\frac{ n+1}{2}}=1$.
		\item If $n$ is even, then $w_i=0$ for all $i\ne \frac{n}{2}$ and $i\ne \frac{n+2}{2}$, and $w_{\frac{n}{2}}=w_{\frac{n+2}{2}}=\frac{1}{2}$.
	\end{itemize}
\end{exem}

In addition there are some special types of \textsf{OWA}s: {\bf centered OWA}  or \textsf{cOWA}\cite{Yager2006}.

The \textsf{OWA} functions are defined in terms of a predetermined vector of weights; namely this vector of weights is fixed previously by the user. In the next section we present a generalized form of \textsf{OWA} in order to relax this situation. The vector of weights will be in function of the vector of inputs ${\bf x}=(x_1,\ldots x_n)$. To achieve that we replace, in the OWA expression, the vector of weights by a family of functions,  called  \textbf{weighted functions}.

\section{Weighted functions}

As mentioned, the \textsf{OWA} functions are means with previously fixed weights. In the literature we can find some kind of functions which overcome this situation, by providing variable weights. This functions are called here \textit{weighted functions}. An important example of that is the Mean of Bajraktarevic, presented in \cite{Beliakov2016}, and a particular case called of \textbf{Mixture function}, which are functions that have the following form:

\begin{equation}
M({\bf x})= \frac{\sum\limits_{i=1}^{n}w_i(x_i)\cdot x_i}{\sum\limits_{i=1}^{n}w_i(x_i)}
\end{equation}

Generally, the mixture functions are not aggregation functions, since they do not always satisfy monotonicity, however the references \cite{Pereira2003,Messiar2006,Mesiar2008} provide sufficient conditions to overcome this situation.

\begin{rmk}
	Note at equation 1 that each weight $w_i(x_i)$ is the value of single variable function; namely the weight is the value of a function $w_i$ applied to the $i$-th position of the input vector ${\bf x}=(x_1,\ldots, x_n)$. However, this restriction can be relaxed in order to obtain a weight $w_i({\bf x})$, i.e., the weight is in function of the whole input. This generalization of mixture operators was done by Pereira \cite{Pereira1999,Pereira2000} and the resulting functions were called of {\bf Generalized Mixture Functions (GMF)}. 
	
	Although Pereira has introduced GMFs he did not provided a deep investigation about them. In references \cite{Pereira1999,Pereira2000,Pereira2003} studied only {\bf GMF}'s depending on a single variable, that is, functions of the form $W({\bf x})=\sum\limits_{i=1}^n w_i(x_i)\cdot x_i$. In what follows we provide some results about such {\bf GMF}s; its relation to \textsf{OWA}'s, Mixture Functions and Pre-aggregations. We finally generalize GMF`s to the notion of {\bf Bounded Generalized Mixture Functions (BGMF)} and provide some relations to the notions of monotonicity, directional monotonicity, Weak-dual and Weak-conjugate functions.
\end{rmk}

\subsection{Weighted Averaging Functions}

\hfill\
\begin{defi}
	A finite family of functions  $\Gamma=\{f_i:[0,1]^n\rightarrow [0,1]\ |\ 1\le i\le n\}$ such that $\sum\limits_{i=1}^n f_i({\bf x})=1$ is called {\bf family of weight-functions (FWF)}. A {\bf Generalized Mixture Function} or simply \textsf{GM} associated to a FWF $\Gamma$ is a function $\textsf{GM}_{\Gamma}:[0,1]^n\rightarrow [0,1]$ given by:
	$$\textsf{GM}_{\Gamma}({\bf x})=\sum_{i=1}^{n} f_i({\bf x})\cdot x_i$$
\end{defi}

\begin{exem}\hfill
	\begin{enumerate}
		\item The \textsf{GM} operator associated to $\Gamma=\left\{ f_i({\bf x})=\frac{1}{n}| \ 1\leq i\leq n\right\}$, $\textsf{GM}_{\Gamma}({\bf x})$, is $Arith({\bf x})$;
		\item The function {\it Minimum} can be obtained from $\Gamma=\{ f_i\ |\ 1\le i\le n\}$, where for all ${\bf x}\in [0,1]^n$, $f_{(n)}({\bf x})=1$ and $f_i({\bf x}) = 0$, if $i\ne (n)$, where $(\cdot):\{1,\cdots,n\}\rightarrow \{1,\cdots,n\}$ is a permutation such that $sort({\bf x})=(x_{(1)},\cdots,x_{(n)})$;
		\item Similarly, the function {\it Maximum} is also of type \textsf{GM} with $\Gamma$ dually defined.
	\end{enumerate}
\end{exem}

\begin{teo}
	For any vector of weights ${\bf w}=(w_1,w_2,...,w_n)$, the function $\textsf{OWA}_{\bf w}$ is a \textsf{GM} function.
\end{teo}
\begin{proof}
	Define $f_i({\bf x})= w_{p(i)}$, where $p:\{1,2,\cdots, n\}\longrightarrow \{1,2,\cdots,n\}$ is the inverse permutation of $q(i)=(i)$. Then,
	\begin{eqnarray*}
		\textsf{GM}_{\Gamma}(x_1,\cdots,x_n) & = & \sum\limits_{i=1}^nf_i(x_1,\cdots,x_n)\cdot x_i\\
		& = & \sum\limits_{i=1}^n w_{p(i)}x_i
	\end{eqnarray*}
	but as $p$ is the inverse of $q$, it follows that $q(p(i))=i$, that is, $(p(i))=i$. Thus,
	$$\textsf{GM}_{\Gamma}(x_1,\cdots,x_n)= \sum\limits_{i=1}^n w_{p(i)}x_{(p(i))}$$
	Making the necessary changes in this sum, we have:
	\begin{eqnarray*}
		\textsf{GM}_{\Gamma}(x_1,\cdots,x_n) & = & \sum\limits_{i=1}^n w_ix_{(i)}\\
		& = & \textsf{OWA}_{\bf w}(x_1,\cdots,x_n)
	\end{eqnarray*}
\end{proof}

\begin{exem}
	If ${\bf w}=(0.3,0.4,0.3)$, then for ${\bf x}=(0.1,1.0,0.9)$ we have $x_1=x_{(3)},\ x_2=x_{(1)}$ and $x_3=x_{(2)}$. Thus, $f_1({\bf x})=0.3,\ f_2({\bf x})=0.3$, $f_3({\bf x})=0.4$, and $\textsf{GM}({\bf x})=0.3\cdot 0.1+ 0.3\cdot1.0 + 0.4\cdot 0.9=0.69$.
\end{exem}

\newpage
Note that, by Theorem 1, any \textsf{OWA} function is \textsf{GM}. However, there are \textsf{GM} functions which are not \textsf{OWA}:

\begin{exem}
	Let $\Gamma=\{\sin(x)\cdot y, 1-\sin(x)\cdot y\}$. The respective \textsf{GM} function is $\textsf{GM}(x,y)=(\sin(x)\cdot y)\cdot x+(1-\sin(x)\cdot y)\cdot y$, which is not an \textsf{OWA} function.
\end{exem}

\begin{prop}
	Mixture functions are a particular case of \textsf{GM}.
\end{prop}
\begin{proof}
	A mixture operator $M({\bf x})$ can be see as a \textsf{GM} function, with weight-functions given by $f_i({\bf x})=\frac{w_i(x_i)}{\sum\limits_{i=1}^n w_i(x_i)}$.
\end{proof}

\begin{rmk}
	The \textsf{GM} function at Example 6 cannot be characterized as a mixture function, since $w_1$ is not a function which depends only on variable $x$ and $w_2$ is not a function which depends only on variable $y$.
\end{rmk}

Below, we propose a new generalized form of \textsf{GM}, which is obtained by relaxing the condition $\sum\limits_{i = 1}^nf_i({\bf x}) = 1$ to $\sum\limits_{i = 1}^nf_i({\bf x})\le 1$. This new family is called Bounded Generalized Mixture functions. 

\begin{defi}
	
	A family of {\bf weak weight-functions (wFWF)} is a finite family of functions\linebreak $\Gamma=\left\{f_i:[0,1]^n\rightarrow [0,1] | \ 1\le i\le n\right\}$ such that:
	\begin{enumerate}
		\item[(I)] $\sum\limits_{i=1}^n f_i({\bf x})\le1$, and
		\item[(II)] $\sum\limits_{i=1}^n f_i(1,\cdots,1)=1$, for all $i\in\{1,2,\cdots,n\}$.
	\end{enumerate} 
	
	A {\bf Bounded Ganeralized Mixture} (\textsf{BGM}) operator associated to a wFWF $\Gamma$ is a function $\textsf{BGM}_{\Gamma}:[0,1]^n\rightarrow [0,1]$ given by:
	$$\textsf{BGM}_{\Gamma}({\bf x})=\sum_{i=1}^{n} f_i({\bf x})\cdot x_i$$
\end{defi}

\begin{rmk}
	\hfill
	\begin{enumerate}
		\item Since \textsf{BGM} is a generalized form of \textsf{GM}, then both \textsf{OWA}s and mixture functions are instances of \textsf{BGM}. Moreover, \textsf{GM} functions are \textsf{BGM} operators subject to the condition: 
		\begin{enumerate}
			\item[(III)] $\sum\limits_{i=1}^n f_i({\bf x})=1$, for any ${\bf x} \in [0,1]^n$,
		\end{enumerate}
		\item Let $\Gamma =\left\{ f_i(x,y)=\frac{x}{n}|\ 1\le i\le n\right\}$. Then, $\textsf{BGM}_{\Gamma}({\bf x})=\sum\limits_{i=1}^n \frac{x_i^2}{n}$ is not a \textsf{GM} operator.
	\end{enumerate}
\end{rmk}

\subsection{Properties of \textsf{GM} and \textsf{BGM} functions}\hfill

Although \textsf{GM} and \textsf{BGM} are generalized forms of \textsf{OWA}, we cannot always guarantee that a \textsf{BGM} is an averaging function, although \textsf{GM} functions are averaging function. The next proposition gives us a sufficient condition to achieve that.

\begin{prop}
	If $\Gamma$ is a FWF, then $\textsf{GM}_{\Gamma}$ is averaging, i.e.:
	
	$$Min({\bf x})\le \textsf{GM}_{\Gamma}({\bf x}) \le Max({\bf x})$$
\end{prop}
\begin{proof}
	For all ${\bf x}=(x_1,...,x_n)$,
	$$Min({\bf x})\le x_i\le Max({\bf x}),\ \forall i=1,2,...,n.$$
	So,
	$$\sum\limits_{i=1}^nf_i({\bf x})\cdot Min({\bf x})\le \sum\limits_{i=1}^nf_i({\bf x})\cdot x_i\le \sum\limits_{i=1}^nf_i({\bf x})\cdot Max({\bf x}),$$
	but as $\sum\limits_{i=1}^nf_i({\bf x})= 1$, it follows that
	$$Min({\bf x})\le \sum\limits_{i=1}^nf_i({\bf x})\cdot x_i \le Max({\bf x})$$
\end{proof}

\begin{rmk} Note that FWF cannot simply be replaced by wFWF, since for $f_1(x,y)=\frac{x}{2}$ and $f_2(x,y)=\frac{y}{2}$, we have  $\textsf{BGM}(0.5, 0.5)=0.25<Min(0.5, 0.5)$, although we can guarantee that:
	$$\textsf{BGM}_{\Gamma}({\bf x})\le Max({\bf x})$$
\end{rmk}

\begin{prop}
	Let $\Gamma$ be a wFWF. Then, the $\textsf{BGM}_{\Gamma}$ is idempotent if, and only, if $\sum\limits_{i=1}^n f_i(x,\cdots,x)=1$ for any $x\in[0,1]$.
\end{prop}

\begin{proof}
	If $\sum\limits_{i=1}^n f_i({\bf x})=1$ and ${\bf x}=(x,...,x)$, then:
	$$\textsf{BGM}_{\Gamma}({\bf x})=\sum\limits_{i=1}^nf_i({\bf x})\cdot x=x\cdot\sum\limits_{i=1}^nf_i({\bf x})=x$$
	
	Reciprocally, if $\textsf{BGM}_{\Gamma}$ is a idempotent function and $\sum\limits_{i=1}^n f_i(x,\cdots,x)<1$ for some $x\in[0,1]$ we have to
	$$\textsf{BGM}_{\Gamma}({\bf x})=\sum\limits_{i=1}^nf_i({\bf x})\cdot x<x\cdot 1=x.$$
	Thus, the condition $\sum\limits_{i=1}^n f_i(x,\cdots,x)=1$ cannot be removed.
\end{proof}

\begin{cor}
	Any \textsf{GM} function is idempotent.
\end{cor}

\begin{rmk}
	The \textsf{BGM} function described in Remark 6, is not idempotent, since $BGM(0.5,0.0)=0.25$. Thus, we cannot always guarantee that a \textsf{BGM} is idempotent
\end{rmk}

\begin{prop}
	If $\Gamma$ is a FWF invariant under translations\footnote{This means that $f_i(x_1+\lambda,x_2+\lambda,...,x_n+\lambda)=f_i(x_1,x_2,...,\linebreak x_n)$ for any ${\bf x}\in [0,1]^n$,  for $i\in\{1,2,\cdots,n\}$ and $\lambda\in[-1,1]$ such that $(x_1+\lambda,x_2+\lambda,...,x_n+\lambda)\in[0,1]^n$.}, then $\textsf{GM}_{\Gamma}$ is shift-invariant.
\end{prop}
\begin{proof}
	Let ${\bf x}=(x_1,...,x_n)\in[0,1]^n$ and $\lambda\in[-1,1]$ such that ${\bf y}=(x_1+\lambda,x_2+\lambda,...,x_n+\lambda)\in[0,1]^n$. then,
	\begin{eqnarray*}
		\textsf{GM}_{\Gamma}({\bf y})
		& = & \sum\limits_{i=1}^n f_i({\bf y})\cdot(x_i+\lambda)\\
		& = & \sum\limits_{i=1}^nf_i({\bf y})\cdot x_i +\ \sum\limits_{i=1}^nf_i({\bf y})\cdot\lambda\\
		& = & \sum\limits_{i=1}^nf_i({\bf x})\cdot x_i+\lambda\\
		& = & \textsf{GM}_{\Gamma}({\bf x})+\lambda
	\end{eqnarray*}
\end{proof}

\begin{rmk}
	The condition FWF is also important to preserve shift-invariance, since if we define $f_1(x,y)=f_2(x,y)=\frac{|x-y|}{2}$, for $(x,y)\ne (1,1)$, and $f_1(1,1)=f_2(1,1)=\frac{1}{2}$, then $f_1$ and $f_2$ are invariant under translations, but $\textsf{BGM}(0,0.1)=0.005$ and $\textsf{BGM}(0+0.1,0.1+0.1)=0.015\ne 0.005+0.1$.
\end{rmk}

\begin{prop}
	If $\Gamma$ is a wFWF and each $f_i\in \Gamma$ is homogeneus of order $k$, then $\textsf{BGM}_{\Gamma}({\bf x})$ is homogeneous of order $k+1$.
\end{prop}
\begin{proof}
	The case $\lambda=0$ is trivial. Now, to $\lambda \ne 0$ we have:
	\begin{eqnarray*}
		\textsf{BGM}_{\Gamma}(\lambda x_1,...,\lambda x_n) & = & \sum\limits_{i=1}^nf_i(\lambda x_1,...,\lambda x_n)\cdot\lambda x_i\\
		& = & \lambda\cdot \sum\limits_{i=1}^n\lambda^kf_i(x_1,...,x_n)x_i\\
		& = & \lambda^{k+1} \cdot \textsf{BGM}_{\Gamma}(x_1,...,x_n)
	\end{eqnarray*}
\end{proof}

\begin{rmk}
	Note that if $\Gamma$ is a FWF, i.e, $\sum\limits_{i=1}^n f_i({\bf x})=1$, then $f_i$ cannot be homogeneous of order $k> 0$, since
	$$1=\sum\limits_{i=1}^nf_i(\lambda x_1,\cdots,\lambda x_n)=\lambda^k \sum\limits_{i=1}^n f_i({\bf x}) = \lambda^k,$$
	i.e., there are no \textsf{GM}'s homogeneous of order $k>1$. However, if we remove this restriction, then we can have $\Gamma$ with homogeneous $f_i$s of order $k> 0$. For example, $f_i({\bf x})=\frac{x_i}{n}$ is homogeneous of order $1$, and so, according to Proposition 5, $\textsf{BGM}_{\Gamma}$ is homogeneus of order $2$.
\end{rmk}

The next example shows a \textsf{GM} function which is not a mixture operator, since $f_i$ does not depend exclusively from $x_i$. This $\textsf{GM}_{\Gamma}$ is idempotent, homogeneous and shift-invariant, but is not monotonic, since $\textsf{GM}_{\Gamma}(0.5,0.2,0.1)=0.375$ and\linebreak $\textsf{GM}_{\Gamma}(0.5,0.22,0.2)=0.368$.

\begin{exem}
	The family $\Gamma$ defined by
	$$f_i(x_1,...,x_n)=\left\{\begin{array}{ll}
	\frac{1}{n}, & \mbox{ if } x_1=\cdots=x_n=0\\
	\frac{x_i}{\sum\limits_{j=1}^n x_j}, & \mbox{otherwise}
	\end{array}\right.$$
	if a FWF and, $$\textsf{GM}_{\Gamma}({\bf x})=\left\{\begin{array}{ll}
	0, & \mbox{ if } x_1,...,x_n=0\\
	\frac{\sum\limits_{i=1}^n x_i^2}{\sum\limits_{i=1}^n x_i}, & \mbox{otherwise} \end{array}\right.$$
\end{exem}

\begin{prop}
	The $N$-dual\footnote{The $N$-dual of a function $F:[0,1]^n\longrightarrow [0,1]$ is $F^N(x_1,\cdots,x_n)=N(F(N(x_1),\cdots,N(x_n))$, where $N$ is a fuzzy negation, i.e., a decreasing function $N:[0,1]\longrightarrow[0,1]$ with $N(0)=1$ and $N(1)=0$.}, with respect to stantard fuzzy negation\footnote{The standard fuzzy negation if $N(x)=1-x$}, of a \textsf{GM} function is also a \textsf{GM} function.
\end{prop}

\begin{proof}
	If $\Gamma$ is a wFWF, then
	\begin{eqnarray*}
		\textsf{GM}_{\Gamma}^N(x_1,\cdots,x_n)
		& = & 1-\sum\limits_{i=1}^n f_i(1-x_1,\cdots,1-x_n)\cdot (1-x_i)\\
		& = & 1-\sum\limits_{i=1}^n f_i(1-x_1,\cdots,1-x_n)+ \sum\limits_{i=1}^n f_i(1-x_1,\cdots,1-x_n)\cdot x_i\\
		& = & \sum\limits_{i=1}^n f_i(1-x_1,\cdots,1-x_n)\cdot x_i\\
		& = & \sum\limits_{i=1}^n g_i(x_1,\cdots,x_n)\cdot x_i,
	\end{eqnarray*}
	where $g_i(x_1,\cdots,x_n)= f_i(1-x_1,\cdots,1-x_n)$.
\end{proof}

\begin{rmk}\hfill
	\begin{enumerate}
		\item The $N$-dual of a \textsf{BGM} with respect to stantard fuzzy negation it will not be a \textsf{BGM} function, but it will be of the form
		$$\textsf{BGM}_{\Gamma}^N({\bf x}) = \textsf{BGM}_{\Gamma'}({\bf x})+h({\bf x}),$$
		where $h:[0,1]^n\rightarrow[0,1]$ is given by $h({\bf x})=1-\sum\limits_{1}^n f_i({\bf x})$.
		\item The dual of a \textsf{GM} function will not always be a \textsf{GM} function, for example: If $N(x)=1-x^n$, then
		\begin{eqnarray*}
			\textsf{GM}_{\Gamma}^N(x_1,\cdots,x_n) &= & \sum\limits_{i=1}^n f_i(1-x_1^n,\cdots,1-x_n^n)\cdot x_i^n\\
			& = &\sum\limits_{i=1}^n\left( f_i(1-x_1^n,\cdots,1-x_n^n)\cdot x_i^{n-1}\right)\cdot x_i
		\end{eqnarray*}
	and $\Gamma'=\left\{g_i({\bf x})= f_i(1-x_1^n,\cdots,1-x_n^n)\cdot x_i^{n-1}|\ 1\le i\le n \right\}$ not is a FWF, but is a wFWF. So, this $N$-dual is a \textsf{BGM} function.
	\end{enumerate}
\end{rmk}

This provides a motivation to define the {\bf weak dual} of a \textsf{GM} function, as follow:
	
	\begin{defi}
		If $\textsf{GM}_{\Gamma}$ is a generalized mixture function, then the {\bf weak dual} of a $\textsf{GM}_{\Gamma}$ with respect to fuzzy negations $N$ is the function:
		$$\textsf{GM}_{\Gamma}^{wN}(x_1,\cdots,x_n)=\sum\limits_{i=1}^n f_i(N(x_1),\cdots,N(x_n))\cdot x_i$$
	\end{defi}
	
	It is obvious that the weak dual of a \textsf{GM} functions is a \textsf{GM} function. Futhermore, we can define the weak dual of a \textsf{BGM} function, which is also a \textsf{BGM} function.
	
	\begin{exem}
		The weak dual of \textsf{GM} defined in example 6 with respect to fuzzy negation $N(x)=1-x^\alpha$ is
		$$\textsf{GM}_{\Gamma}^{wN}(x_1,\cdots,x_n)=\left\{\begin{array}{ll}0, & \mbox{ if } x_1,\cdots,x_n=0\\ \frac{n-\sum\limits_{i=1}^n (1-x_i)^\alpha \cdot x_i}{n-\sum\limits_{i=1}^n (1-x_i)^\alpha}, & \mbox{ otherwise}\end{array}\right.$$
	\end{exem}
	
	The construction of weak duals can be generalized, as follow:
	
	\begin{prop}
		Let $\gamma_1,\cdots,\gamma_n:[0,1]\longrightarrow[0,1]$ be functions. If $\textsf{GM}_{\Gamma}$ is a \textsf{GM} (or \textsf{BGM}) function, then
		$$\textsf{GM}_{\Gamma}^{\gamma_1,\cdots,\gamma_n}(x_1,\cdots,x_n)=\sum\limits_{i=1}^n f_i(\gamma_1(x_1),\cdots,\gamma_n(x_n))\cdot x_i$$
		is a \textsf{GM} (or \textsf{BGM}) function.
	\end{prop}
	
	\begin{proof}
		Straightforward.
	\end{proof}
	
	When $\gamma_1=\cdots=\gamma_n=\gamma$ are automorphism we say that $\textsf{GM}_{\Gamma}^{\gamma_1,\cdots,\gamma_n}$ is a {\bf weak conjugate} of $\textsf{GM}_{\Gamma}$ and we denote by $\textsf{GM}_{\Gamma}^\gamma$.

\begin{prop}
	If $\Gamma = \{f_1,\cdots,f_n\}$ is a wFWF, then $\Gamma^R=\{g_1,\cdots,f_n\}$, whete $g_i=f_{n-i+1}$, also is a wFWF. Besides that, $\textsf{BGM}_\Gamma^R=\textsf{BGM}_{\Gamma^R}$
\end{prop}

\begin{proof}
	Sraightforward.
\end{proof}

In Pereira \cite{Pereira1999,Pereira2000,Pereira2003} some criteria were considered to generate monotone \textsf{GM}. However, in this work we will not give a deep exposition of monotonicity, instead we present a brief discussion on a more weakened form, called {\bf weak monotonicity} or {\bf directional monotonicity}.

\subsection{Direcional Monotonicity}

There are many $n$-ary functions that do not satisfy the monotonicity although, they are monotone with respect to certain directions. In this sence, Wilkin and Beliakov (in \cite{Wilkin2014,Wilkin2015}) introduceed the concept of {\bf weakly monotonicity}, that was generalized by Bustince {\it et al.} in \cite{Bustince2015}, which define the notion of {\bf diretional monotonicity}.

\begin{defi}
	Let ${\bf r}=(r_1,\cdots,r_n)$ be a nnt null $n$-dimentional vector. A function $F:[0,1]^n\longrightarrow[0,1]$ is ${\bf r}${\bf-increasing} if for all ${\bf x}=(x_1,\cdots,x_n)$ and $t>0$ such that $(x_1+tr_1,\cdots,x_n+tr_n)\in[0,1]^n$, we have
	$$F(x_1,\cdots,x_n)\le F(x_1+tr_1,\cdots,x_n+tr_n),$$
	that is, $F$ is increasing in the direction of vector $\mathbf{r}$.
\end{defi}

\begin{defi}
	A function $F:[0,1]^n\longrightarrow[0,1]$ is an $n$-ary {\bf pre-aggregation} function (or simply pre-aggregation) if satisfies the boundary condition, $F(0,\cdots,0)=0$ and $F(1,\cdots,1)=1$, and is ${\bf r}$-increasing for some direction ${\bf r}\in[0,1]^n$.
\end{defi}

Lucca {\it et al.} \cite{Luca2016} presented properties, constructions and application for pre-aggregations. They show that the following functions are pre-aggregations:

\begin{exem}\hfill
	\begin{enumerate}
		\item $Mode(x_1,\cdots,x_n)$ is $(1,1)$-increasing;
		\item $F(x,y)=x-(max\{0,x-y\})^2$ is $(0,1)$-increasing;
		\item The weighted Lehmer mean (with convention $0/0=0$) $$L_{\lambda}(x,y)=\frac{\lambda x^2+(1-\lambda)y^2}{\lambda x+(1-\lambda)y}, \mbox{ where } 0<\lambda<1$$
		is $(1-\lambda,\lambda)$-increasing;
		\item $$A(x,y)=\left\{\begin{array}{ll}x(1-x), & \mbox{ if }y\le 3/4\\ 1, & \mbox{ otherwise}\end{array}\right.$$
		is $(0,a)$-increasing for any $a>0$, but for no other direction;
		\item $$B(x,y)= \left\{\begin{array}{ll}y(1-y), & \mbox{ if }x\le 3/4\\ 1, & \mbox{ otherwise}\end{array}\right.$$
		is $(b,0)$-increasing for any $b>0$, but for no other direction.
	\end{enumerate}
\end{exem}

\begin{rmk}
	Any aggregation functions is also a pre-aggregation function.
\end{rmk}

\begin{prop}
	If $\textsf{BGM}_{\Gamma}$ is shift-invariant, is a pre-aggregation $(k,\cdots,k)$-increasing.
\end{prop}
\begin{proof}
	Note that for all ${\bf x}=(x_1,x_2,\cdots,x_n)\in[0,1]^n$ and any $t>0$ such that $(x_1+tk,x_2+tk,\cdots,x_n+tk)\in[0,1]$ we have 
	$$\textsf{BGM}_{\Gamma}(x_1+tk,\cdots,x_n+tk)=\textsf{BGM}_{\Gamma}(x_1,\cdots,x_n)+tk,$$
	and so
	$$\textsf{BGM}_{\Gamma}(x_1,\cdots,x_n)\le \textsf{BGM}(x_1+tk,\cdots,x_n+tk)$$	
\end{proof}

\begin{cor}
	If $\Gamma$ is a FWF invariant under translations, i.e, $f_i(x_1+\lambda,x_2+\lambda,...,x_n+\lambda)=f_i(x_1,x_2,...,x_n)$,  for $i\in\{1,2,\cdots,n\}$, for any ${\bf x}=(x_1,\cdots,x_n)\in [0,1]^n$ and $\lambda\in[0,1]$ such that $(x_1+\lambda,x_2+\lambda,...,x_n+\lambda)\in[0,1]^n$, then $\textsf{BGM}_{\Gamma}$ is a pre-aggregation $(k,\cdots,k)$-increasing.
\end{cor}
\begin{proof}
	By Proposition 5 $\textsf{BGM}_{\Gamma}$ is shift-invariant, and so, by Proposition 9, $\textsf{BGM}_{\Gamma}$ is a pre-aggregation function $(k,k,\cdots,k)$-increasing.
\end{proof}

In fact, the conditions required by Corollary 2 are very strong. In the following proposition, we relax these conditions:

\begin{prop}
	If $\Gamma$ is a rFWF with $f_i(x_1,\cdots,x_n)\le f_i(x_1+\lambda,\cdots,x_i+\lambda)$, for $i\in\{1,2,\cdots,n\}$, for any ${\bf x}=(x_1,\cdots,x_n)\in [0,1]^n$ and $\lambda\in[0,1]$ such that $(x_1+\lambda,x_2+\lambda,...,x_n+\lambda)\in[0,1]^n$, then $\textsf{BGM}_{\Gamma}$ is a pre-aggregation function $(k,k,\cdots,k)$-increasing.
\end{prop}
\begin{proof}
	For any ${\bf x}=(x_1,\cdots,x_n)\in [0,1]^n$ and $\lambda\in[0,1]$ such that ${\bf y}=(x_1+\lambda,x_2+\lambda,...,x_n+\lambda)\in[0,1]^n$ we observe that
	\begin{eqnarray*}
		\textsf{BGM}_{\Gamma}({\bf y}) &=& \sum\limits_{i=1}^n f_i({\bf y})\cdot(x_i+\lambda)\\
		& = & \sum\limits_{i=1}^nf_i({\bf y})\cdot x_i + \sum\limits_{i=1}^nf_i({\bf y})\cdot\lambda\\
		& \geq & \sum\limits_{i=1}^nf_i({\bf x})\cdot x_i+\lambda\\
		& \geq & \textsf{BGM}_{\Gamma}({\bf x})
	\end{eqnarray*}
\end{proof}

\begin{exem}
	Let $\Gamma$be the family of functions: $$f_i(x_1,\cdots,x_n)=\left\{\begin{array}{ll}\frac{1}{n},& \mbox{ if } x_1=\cdots=x_n\\ \frac{x_{(1)}-x_i}{\sum\limits_{j=1}^n(x_{(1)}-x_j)},& \mbox{ otherwise}\end{array} \right..$$
	We can easily prove that all those functions satisfy:
	$$f_i(x_1+\lambda,x_2+\lambda,\cdots,x_n+\lambda)= f_i(x_1,x_2,\cdots,x_n).$$
	More generally, for any $\alpha\geq 1$
	$$f_i(x_1,\cdots,x_n)=\left\{\begin{array}{ll}\frac{1}{n},& \mbox{ if } x_1=\cdots=x_n\\ \frac{x_{(1)}-x_i}{\sum\limits_{j=1}^n(x_{(1)}-x_j)^\alpha},& \mbox{ otherwise}\end{array} \right.$$ is such that
	$$f_i(x_1,x_2,\cdots,x_n)\le f_i(x_1+\lambda,x_2+\lambda,\cdots,x_n+\lambda).$$
	Thus, the corresponding \textsf{BGM} is $(k,\cdots,k)$-increasing. In additon, note that, for $\alpha>1$, $\Gamma=\{f_i\}$ does not satisfies $\sum\limits_{i=1}^n f_i(\bf x)=1$.
\end{exem}

We can also establish a criterion analogous to Proposition 10, by replacing the vector $(k,\cdots, k)$ by direction ${\bf r}$, as follow:

\begin{prop}
	If $\Gamma$ is a FWF such that there is a diretional vector ${\bf r}=(r_1,r_2,\cdots,r_n)\in[0,1]^n$ with $f_i(x_1,\cdots,x_n)\le f_i(x_1+\lambda\cdot r_1,\cdots,x_i+\lambda\cdot r_n)$, for $i\in\{1,2,\cdots,n\}$, for any ${\bf x}=(x_1,\cdots,x_n)\in [0,1]^n$ and $\lambda\in[0,1]$ such that $(x_1+\lambda\cdot r_1,x_2+\lambda\cdot r_2,...,x_n+\lambda\cdot r_n)\in[0,1]^n$, then $\textsf{BGM}_{\Gamma}$ is a pre-aggregation function ${\bf r}$-increasing.
\end{prop}
\begin{proof}
	Similar to Proposition 10.
\end{proof}

\begin{cor}
	If $\Gamma$ is a wFWF such that there is a diretional vector ${\bf r}$ with $\frac{\partial{f_i}}{\partial {\bf r}}({\bf x})\geq0$ for any $f_i\in\Gamma$ and ${\bf x}\in[0,1]^n$, then $\textsf{BGM}_{\Gamma}$ is a pre-aggregation function ${\bf r}$-increasing.
\end{cor}

\begin{exem}
	If $f_i=w_i$ is constant, then $\textsf{BGM}_{\Gamma}$ is ${\bf r}$-increaing for any direction ${\bf r}$. Now, given a direction ${\bf r}=(r_1,\cdots,r_n)\in[0,1]^n$ we can build a ${\bf r}$-increasing $\textsf{BGM}$ function defining:
	$$f_i(x_1,\cdots,x_n)=\left\{\begin{array}{ll} 0,& \mbox{ if } x_{(n)}=0\\ \frac{min\left\{\frac{x_i}{r_i},1\right\}}{n}, & \mbox{ otherwise}\end{array}\right.$$
\end{exem}

\newpage
In what follows we propose a new GM function which will be investigated and applied in this paper.

\begin{exem}
	Definition: Let $\Gamma$ be the following family of functions:
	 $$f_i({\bf x})=\left\{\begin{array}{ll}
	\frac{1}{n}, & \mbox{if } {\bf x}=(x,...,x)\\ \frac{1}{n-1}\left(1-\frac{|x_i-Med({\bf x})|}{\sum\limits_{j=1}^n|x_j-Med({\bf x})|}\right), & \mbox{otherwise}\end{array}\right.$$
	$\Gamma$ is a FWF, whose \textsf{GM} function, denoted by $\mathbf{H}$. The computation of $\mathbf{H}$ can be performed using the following expressions:
	{\small \begin{eqnarray*}
			\mathbf{H} ({\bf x}) & = & \left\{ \begin{array}{ll}
				x, & \mbox{ if } {\bf x}=(x,...,x)\\
				\frac{1}{n-1}\sum\limits_{i=1}^n\left(x_i-\frac{x_i|x_i-Med({\bf x})|}{\sum\limits_{j=1}^n|x_j-Med({\bf x})|} \right), &  \mbox{otherwise}\end{array}\right.
	\end{eqnarray*}}
\end{exem}

An interesting property that function is that:
$$\mathbf{H}({\bf x})=\textsf{OWA}_{(f_1({\bf x}),\cdots, f_n({\bf x}))}({\bf x})$$

In the next subsection we discuss others properties of the function $\mathbf{H}$.

\subsection{Properties of $\mathbf{H}$}\hfill

In this part of paper we will discuss about the properties of operator \textbf{H}. It is easy to check that $\sum\limits_{i=1}^n f_i({\bf x})=1$ for any ${\bf x}\in[0,1]^n$ and therefore, $\mathbf{H}$ is an idempotent averaging function. Furthermore, its weight-functions are invariant under translations and is also homogeneous of order $0$, because: 

\begin{enumerate}
	\item For ${\bf y}=(x_1+\lambda,...,x_n+\lambda)$ we have $Med({\bf x'})=Med({\bf x})+\lambda$ and for ${\bf x}\ne(x,...,x)$. So,
	\begin{eqnarray*}
		f_i({\bf y}) & = & \textstyle\frac{1}{n-1}\left(1-\frac{|x_i+\lambda-Med({\bf y})|}{\sum\limits_{j=1}^n|x_j+\lambda-Med({\bf x'})|}\right)\\
		& = & \textstyle\frac{1}{n-1}\left(1-\frac{|x_i+\lambda-(Med({\bf x})+\lambda)|}{\sum\limits_{j=1}^n|x_j+\lambda-(Med({\bf x})+\lambda)|}\right)\\
		& = & \textstyle\frac{1}{n-1}\left(1-\frac{|x_i-Med({\bf x})|}{\sum\limits_{j=1}^n|x_j-Med({\bf x})|}\right)\\
		& = & f_i({\bf x}).
	\end{eqnarray*}
	Therefore, $(f_1({\bf y}),...,f_n({\bf y}))=(f_1({\bf x}),...,f_n({\bf x}))$. 
	
	\item The case in which ${\bf x} =(x,...,x)$ is immediate.
	
	\item To check the second property, make ${\bf y}=(\lambda x_1,...,\lambda x_n)$, note that $Med({\bf y})=\lambda Med ({\bf x})$ and for ${\bf x} \ne (x,...,x)$	
	\begin{eqnarray*}
		f_i({\bf y}) & = & \textstyle\frac{1}{n-1}\left(1-\frac{|\lambda x_i-Med({\bf \lambda x})|}{\sum\limits_{j=1}^n|\lambda x_j-Med({\bf \lambda x})|}\right)\\
		& = & \textstyle\frac{1}{n-1}\left(1-\frac{|\lambda x_i-\lambda Med({\bf x})|}{\sum\limits_{j=1}^n|\lambda x_j-\lambda Med({\bf x})|}\right)\\
		& = & \textstyle\frac{1}{n-1}\left(1-\frac{|\lambda|\cdot|x_i-Med({\bf x})|}{|\lambda|\cdot\sum\limits_{j=1}^n|x_j-Med({\bf x})|}\right)\\
		%		& = & \textstyle\frac{1}{n-1}\left(1-\frac{|x_i-Med({\bf x})|}{\sum\limits_{j=1}^n|x_j-Med({\bf x})|}\right)\\
		& = & f_i({\bf x})
	\end{eqnarray*}
	Therefore,$(f_1({\bf x''}),...,f_n({\bf x''}))=(f_1({\bf x}),..., f_n({\bf x}))=f({\bf x})$.
	\item The case in which ${\bf x} =(x,...,x)$ is also immediately. Note that the case in which $\lambda=0$ is obvious.
\end{enumerate}

\begin{cor}
	$\mathbf{H}$ is shift-invariant and homogeneous.  
\end{cor}

In addition to idempotency, homogeneity and shift-invariance $\mathbf{H}$ has the following proprerties.

\begin{prop}
	$\mathbf{H}$ has no neutral element.
\end{prop}
\begin{proof}
	Suppose $\mathbf{H}$ has a neutral element $e$, find the vector of weight for ${\bf x}=(e,...,e,x,e,...,e)$. Note that if $n\geq 3$, then $Med({\bf x})=e$ and therefore,
	\begin{eqnarray*}
		f_i({\bf x}) & = & \textstyle\frac{1}{n-1}\left( 1-\frac{|x_i-Med({\bf x})|}{\sum\limits_{j=1}^n|x_j-Med({\bf x})|}\right)\\
		& = & \textstyle\frac{1}{n-1}\left(1-\frac{|x_i-e|}{\sum\limits_{j=1}^n|x_j-e|}\right)\\
		& = & \textstyle\frac{1}{n-1}\left(1-\frac{|x_i-e|}{|x-e|}\right)
	\end{eqnarray*}
	therefore,
	$$f_i({\bf x})=	\displaystyle\left\{\begin{array}{ll}
	\frac{1}{n-1}, \mbox{ if } x_i=e\\
	0, \mbox{ if } x_i=x
	\end{array}\right., \mbox{ to } n\geq 3$$
	i.e.,
	$$f({\bf x})=\textstyle\left(\frac{1}{n-1},...,\frac{1}{n-1},0,\frac{1}{n-1},...,\frac{1}{n-1} \right)$$
	and
	$$\mathbf{H}({\bf x})=(n-1)\cdot\frac{e}{n-1}=e$$
	But since $e$ is a neutral element of $\mathbf{H}$, $\mathbf{H}({\bf x})=x$. Absurd, since we can always take  $x\ne e$.
	
	For $n=2$, we have $Med({\bf x})=\frac{x+e}{2}$, where ${\bf x}=(x,e)$ or ${\bf x}=(e,x)$. In both cases it is not difficult to show that $f({\bf x})=(0.5,0.5)$ and $\mathbf{H}({\bf x})=\frac{x+e}{2}$. Thus, taking $x\ne e$, again we have $\mathbf{H}(x,e)\ne x$.
\end{proof}

\begin{prop}
	$\mathbf{H}$ has no absorbing elements.
\end{prop}
\begin{proof}
	To $n=2$, we have $\mathbf{H}({\bf x})=\frac{x_1+x_2}{2}$, which has no absorbing elements. Now for $n\geq 3$ we have to ${\bf x}=(a,0,...,0)$ with $Med({\bf x})=0$ therefore,
	$$f_1({\bf x})=\frac{1}{n-1}\left(1-\frac{a}{a} \right)=0 \mbox{ and } f_i=\frac{1}{n-1}, \forall i=2,...,n.$$
	therefore,
	$$\mathbf{H}(a,0,...,0)=0\cdot a+\frac{1}{n-1}\cdot 0+...+\frac{1}{n-1}\cdot 0=a\Rightarrow a=0,$$
	but to ${\bf x}=(a,1,...,1)$ we have to $Med({\bf x})=1$. Furthermore,
	$$f_1({\bf x})=\frac{1}{n-1}\left(1-\frac{1-a}1-{a} \right)=0 $$
	and
	$$f_i=\frac{1}{n-1} \mbox{ para } i=2,3,...,n.$$
	therefore,
	$$\mathbf{H}(a,1,...,1)=0\cdot a+\frac{1}{n-1}\cdot 1+...+\frac{1}{n-1}\cdot 1=a\Rightarrow a=1.$$
	With this we prove that $\mathbf{H}$ does not have absorbing elements.
\end{proof}

\begin{prop}
	$\mathbf{H}$ has no zero divisors.
\end{prop}
\begin{proof}
	Let $a\in\ ]0,1[$ and consider ${\bf x}=(a,x_2,...,x_n)\in\ ]0,1]^n$. In order to have $\mathbf{H}({\bf x})=\sum\limits_{i=1}^n f_i({\bf x})\cdot x_i=0$ we have $f_i({\bf x})\cdot x_i=0$ for all $i=1,2,...,n$. But as $a\ne 0$ and we can always take $x_2,x_3,...,x_n$ also different from zero, then for each $i=1,2,...,n$ there remains only the possibility of terms:
	$$f_i{(\bf x)}=0 \mbox{ para } i=1,2,...,n.$$
	This is an absurd, for $f_i({\bf x})\in [0,1]$ e $\sum\limits_{i=1}^nf_i({\bf x})=1$. like this, $\mathbf{H}$ has no zero divisors.
\end{proof}

\begin{prop}
	$\mathbf{H}$ does not have one divisors
\end{prop}
\begin{proof}
	Just to see that $a\in\ ]0,1[$, we have to $\mathbf{H}(a,0,...,0)=f_1({\bf x}).a\le a<1$.
\end{proof}

\begin{prop}
	$\mathbf{H}$ is symmetric.
\end{prop}

\begin{proof}
	Let $P:\{1,2,...,n\}\rightarrow\{1,2,...,n\}$ be a permutation. So we can easily see that $$Med(x_{P(1)},x_{P(2)},...,x_{P(n)})=Med(x_1,x_2,...,x_n)$$ for all ${\bf x}=(x_1,x_2,...,x_n)\in[0,1]^n$. We also have to $\sum\limits_{i=1}^n |x_{P(i)}-Med(x_{P(1)},x_{P(2)},...,x_{P(n)})|=\sum\limits_{i=1}^n|x_i-Med({\bf x})|$. Thus, it suffices to consider the case where ${\bf y}=(x_{P(1)},x_{P(2)},...,x_{P(n)})\ne (x,x,...,x)$. But for ${\bf y}\ne (x,x,...,x)$ we have to:
	{\small \begin{eqnarray*}
			\mathbf{H} ({\bf y}) & = &	\textstyle\frac{1}{n-1}\sum\limits_{i=1}^n\left(x_{P(i)}-\frac{x_{P(i)}|x_{P(i)}-Med({\bf y})|}{\sum\limits_{j=1}^n|x_{P(i)}-Med({\bf y})|} \right)\\
			& = &\textstyle\frac{\sum\limits_{i=1}^n x_{P(i)}}{n-1}-\frac{1}{n-1}\cdot\sum\limits_{i=1}^n \frac{x_{P(i)}|x_{P(i)}-Med({\bf x})|}{\sum\limits_{j=1}^n|x_{P(i)}-Med({\bf x})|}\\
			& = & \textstyle\frac{\sum\limits_{i=1}^n x_i}{n-1}-\frac{1}{n-1}\cdot\sum\limits_{i=1}^n \frac{x_{P(i)}|x_{P(i)}-Med({\bf x})|}{\sum\limits_{j=1}^n|x_i-Med({\bf x})|}\\
			& = &\textstyle\frac{\sum\limits_{i=1}^n x_i}{n-1}-\frac{1}{n-1}\cdot\sum\limits_{i=1}^n \frac{x_i|x_i-Med({\bf x})|}{\sum\limits_{j=1}^n|x_i-Med({\bf x})|}\\
			& = & \mathbf{H}({\bf x}).
	\end{eqnarray*}}
\end{proof}

\begin{prop}
	If $N:[0,1]\longrightarrow[0,1]$ is the standard fuzzy negation, then $\mathbf{H}^N=\mathbf{H}$.
\end{prop}

\begin{proof}
	If ${\bf x}=(x,\cdots,x)$, then 
	\begin{eqnarray*}
		\mathbf{H}^N({\bf x}) & = & 1-\mathbf{H}(1-x,1-x,\cdots,1-x)\\
		& = & 1-(1-x)=x=\mathbf{H}({\bf x})
	\end{eqnarray*}
	
	For ${\bf x}\ne (x,\cdots,x)$ and ${\bf y}=(1-x_1,\cdots,1-x_n)$, we have:	
		\begin{eqnarray*}
			\mathbf{H}^N({\bf x}) & = &  1- \textstyle\frac{1}{n-1}\sum\limits_{i=1}^n\left(1-x_i -\frac{(1-x_i)|1-x_i- Med({\bf y})|}{\sum\limits_{j=1}^n|1-x_i - Med({\bf y})|}\right)\\
			& = & 1- \textstyle\frac{1}{n-1}\sum\limits_{i=1}^n\left( 1-x_i -\frac{(1-x_i)|1-x_i- 1 +Med({\bf x})|}{\sum\limits_{j=1}^n|1-x_i - 1+ Med({\bf x})|}\right)\\
			& = &  1- \textstyle\frac{1}{n-1}\sum\limits_{i=1}^n\left( 1-x_i -\frac{(1-x_i)|-x_i+Med({\bf x})|}{\sum\limits_{j=1}^n|-x_i+ Med({\bf x})|}\right)\\
			& = &  1- \textstyle\frac{1}{n-1}\sum\limits_{i=1}^n\left( 1-x_i -\frac{(1-x_i)|x_i-Med({\bf x})|}{\sum\limits_{j=1}^n|x_i-Med({\bf x})|}\right)\\
			& = &  1- \textstyle\frac{1}{n-1}\left[n -\sum\limits_{i=1}^n\left(x_i -\frac{x_i|x_i-Med({\bf x})|}{\sum\limits_{j=1}^n|x_i-Med({\bf x})|}\right)-\textstyle\sum\limits_{i=1}^n\frac{|x_i-Med({\bf x})|}{\sum\limits_{j=1}^n|x_i-Med({\bf x})|}\right]\\
			& = & 1- \textstyle\frac{1}{n-1}\left[n -1 -\sum\limits_{i=1}^n\left(x_i -\frac{x_i|x_i-Med({\bf x})|}{\sum\limits_{j=1}^n|x_i-Med({\bf x})|}\right)\right]\\
			& = &   \textstyle\frac{1}{n-1}\sum\limits_{i=1}^n\left(x_i -\frac{x_i|x_i-Med({\bf x})|}{\sum\limits_{j=1}^n|x_i-Med({\bf x})|}\right)\\
			& = & \mathbf{H}({\bf x})
	\end{eqnarray*}
\end{proof}

\begin{prop}
	If $k>0$, then $\mathbf{H}$ is $(k,\cdots,k)$-increasing.
\end{prop}

\begin{proof}
	As $\mathbf{H}$ is shift-invariant, its follow that $\mathbf{H}$ is $(k,\cdots,k)$-increasing.
\end{proof}

\begin{cor}
	$\mathbf{H}$ is a pre-aggregation function.
\end{cor}

Therefore, $\mathbf{H}$ satisfies the following properties:

\begin{itemize}
	\item Idempotency;
	\item Homogeneity;
	\item Shift-invariance;
	\item Symmetry;
	\item has no neutral element;
	\item has no absorbing elements;
	\item has no zero divisors;
	\item does not have one divisors;
	\item is self dual;
	\item is a preagregation $(k,\cdots,k)$-increasing.
\end{itemize}

Aggregation functions are very important for Computer Science, since in many applications the expected result is a single data, and therefore they usually use aggregation functions to convert this set of data into a unique output. In fact, pre-aggregation can also be applied. In this sense, the Appendix contains a simple application which apply GM functions on the problem of image reduction.

\section{Final remarks}

In this paper we study two generalized forms of Ordered Weighted Averaging function and Mixture function, called \textbf{Generalized Mixture} and \textbf{Bounded Generalized Mixture} functions. This functions are defined by weights, which are obtained dynamically from of each input vector ${\bf x}\in [0,1]^n$. We demonstrated, among other results, that \textsf{OWA} and mixture functions are particular cases of \textsf{GM} and \textsf{BGM} functions, and thus functions likes {\it Arithmetic Mean}, {\it Median}, {\it Maximum}, {\it Minimum} and \textsf{cOWA} are also instances of \textsf{GM} function.

In the second part of this work, we present some properties as well as constructs and examples of GM functions. In particular we define a special \textsf{GM} function, called $\mathbf{H}$. We show that $\mathbf{H}$ satisfies important properties like: Idempotency, symmetry, homogeneity, shift-invariance; it does not have neither zero and one divisors nor neutral elements. We further prove that $\mathbf{H}$ is a pre-aggregation $(k,\cdots,k)$-increasing.

A illustrative application is presented in the Appendix. 
A further insight into the applications of these functions will be addressed in future works.

\section*{References}

\bibliographystyle{elsarticle-num}
\bibliography{refs}

\begin{thebibliography}{34}
\expandafter\ifx\csname natexlab\endcsname\relax\def\natexlab#1{#1}\fi
\providecommand{\url}[1]{\texttt{#1}}
\providecommand{\href}[2]{#2}
\providecommand{\path}[1]{#1}
\providecommand{\DOIprefix}{doi:}
\providecommand{\ArXivprefix}{arXiv:}
\providecommand{\URLprefix}{URL: }
\providecommand{\Pubmedprefix}{pmid:}
\providecommand{\doi}[1]{\href{http://dx.doi.org/#1}{\path{#1}}}
\providecommand{\Pubmed}[1]{\href{pmid:#1}{\path{#1}}}
\providecommand{\bibinfo}[2]{#2}
\ifx\xfnm\relax \def\xfnm[#1]{\unskip,\space#1}\fi
%Type = Article
\bibitem[{Dubois and Prade(2004)}]{Dubois2004}
\bibinfo{author}{D.~Dubois}, \bibinfo{author}{H.~Prade},
\newblock \bibinfo{title}{On the use of aggregation operations in information
  fusion processes},
\newblock \bibinfo{journal}{Fuzzy Sets and Systems} \bibinfo{volume}{142}
  (\bibinfo{year}{2004}) \bibinfo{pages}{143 -- 161}.
  \bibinfo{note}{Aggregation Techniques}.
%Type = Article
\bibitem[{Farias et~al.(2016)Farias, Lopes, Bedregal, and
  Santiago}]{FariasJIFS2016}
\bibinfo{author}{A.~D.~S. Farias}, \bibinfo{author}{L.~R.~A. Lopes},
  \bibinfo{author}{B.~C. Bedregal}, \bibinfo{author}{R.~H.~N. Santiago},
\newblock \bibinfo{title}{Closure properties for fuzzy recursively enumerable
  languages and fuzzy recursive languages,},
\newblock \bibinfo{journal}{Journal of Intelligent $\&$ Fuzzy Systems}
  \bibinfo{volume}{31} (\bibinfo{year}{2016}) \bibinfo{pages}{1795--1806}.
%Type = Article
\bibitem[{Paternain et~al.(2012)Paternain, Jurio, Barrenechea, Bustince,
  Bedregal, and Szmidt}]{Paternain2012}
\bibinfo{author}{D.~Paternain}, \bibinfo{author}{A.~Jurio},
  \bibinfo{author}{E.~Barrenechea}, \bibinfo{author}{H.~Bustince},
  \bibinfo{author}{B.~Bedregal}, \bibinfo{author}{E.~Szmidt},
\newblock \bibinfo{title}{An alternative to fuzzy methods in decision-making
  problems},
\newblock \bibinfo{journal}{Expert Systems with Applications}
  \bibinfo{volume}{39} (\bibinfo{year}{2012}) \bibinfo{pages}{7729 -- 7735}.
%Type = Article
\bibitem[{Paternain et~al.(2015)Paternain, Fernandez, Bustince, Mesiar, and
  Beliakov}]{Paternain2015}
\bibinfo{author}{D.~Paternain}, \bibinfo{author}{J.~Fernandez},
  \bibinfo{author}{H.~Bustince}, \bibinfo{author}{R.~Mesiar},
  \bibinfo{author}{G.~Beliakov},
\newblock \bibinfo{title}{Construction of image reduction operators using
  averaging aggregation functions},
\newblock \bibinfo{journal}{Fuzzy Sets and Systems} \bibinfo{volume}{261}
  (\bibinfo{year}{2015}) \bibinfo{pages}{87 -- 111}. \bibinfo{note}{Theme:
  Aggregation operators}.
%Type = Book
\bibitem[{Beliakov et~al.(2016)Beliakov, Bustince, and Calvo}]{Gleb2016}
\bibinfo{author}{G.~Beliakov}, \bibinfo{author}{H.~Bustince},
  \bibinfo{author}{T.~Calvo}, \bibinfo{title}{A Practical Guide to Averaging
  Functions}, volume \bibinfo{volume}{329} of \textit{\bibinfo{series}{Studies
  in Fuzziness and Soft Computing}}, \bibinfo{publisher}{Springer},
  \bibinfo{year}{2016}.
%Type = Article
\bibitem[{Yager(1988)}]{Yager1988}
\bibinfo{author}{R.~R. Yager},
\newblock \bibinfo{title}{Ordered weighted averaging aggregation operators in
  multicriteria decision making},
\newblock \bibinfo{journal}{IEEE Transactions on Systems, Man, and Cybernetics}
  \bibinfo{volume}{18} (\bibinfo{year}{1988}) \bibinfo{pages}{183 -- 190}.
%Type = Inproceedings
\bibitem[{Pereira and Pasi(1999)}]{Pereira1999}
\bibinfo{author}{R.~A.~M. Pereira}, \bibinfo{author}{G.~Pasi},
\newblock \bibinfo{title}{On non-monotonic aggregation: mixture operators},
\newblock in: \bibinfo{booktitle}{Proc. 4th Meeting of the EURO Working Group
  on Fuzzy Sets (EUROFUSE'99) and 2nd Internat. Conf. on Soft and Intelingent
  Computing (SIC'99)}, \bibinfo{address}{Budapest, Hungary},
  \bibinfo{year}{1999}.
%Type = Inproceedings
\bibitem[{Pereira(2000)}]{Pereira2000}
\bibinfo{author}{R.~A.~M. Pereira},
\newblock \bibinfo{title}{The orness of mixture operators: the exponential
  case},
\newblock in: \bibinfo{booktitle}{Proc. 8th Internat. Conf. on Information
  Processing and Management of Uncertainty in Knowledge-Based Systems
  (IPMU'200)}, \bibinfo{address}{Madrid, Spain}, \bibinfo{year}{2000}.
%Type = Article
\bibitem[{Pereira and Ribeiro(2003)}]{Pereira2003}
\bibinfo{author}{R.~A.~M. Pereira}, \bibinfo{author}{R.~A. Ribeiro},
\newblock \bibinfo{title}{Aggregation with generalized mixture operators using
  weighting functions},
\newblock \bibinfo{journal}{Fuzzy Sets and Systems} \bibinfo{volume}{137}
  (\bibinfo{year}{2003}) \bibinfo{pages}{43 -- 58}. \bibinfo{note}{Preference
  Modelling and Applications}.
%Type = Article
\bibitem[{Yager(2006)}]{Yager2006}
\bibinfo{author}{R.~R. Yager},
\newblock \bibinfo{title}{Centered {OWA} operators},
\newblock \bibinfo{journal}{Soft Computing} \bibinfo{volume}{11}
  (\bibinfo{year}{2006}) \bibinfo{pages}{631--639}.
%Type = Article
\bibitem[{Lizasoain and Moreno(2013)}]{LIZASOAIN2013}
\bibinfo{author}{I.~Lizasoain}, \bibinfo{author}{C.~Moreno},
\newblock \bibinfo{title}{Owa operators defined on complete lattices},
\newblock \bibinfo{journal}{Fuzzy Sets and Systems} \bibinfo{volume}{224}
  (\bibinfo{year}{2013}) \bibinfo{pages}{36 -- 52}. \bibinfo{note}{Theme:
  Aggregation functions and implications}.
%Type = Article
\bibitem[{Llamazares(2015)}]{LLAMAZARES2015131}
\bibinfo{author}{B.~Llamazares},
\newblock \bibinfo{title}{Constructing choquet integral-based operators that
  generalize weighted means and owa operators},
\newblock \bibinfo{journal}{Information Fusion} \bibinfo{volume}{23}
  (\bibinfo{year}{2015}) \bibinfo{pages}{131 -- 138}.
%Type = Article
\bibitem[{Xu(2006)}]{XU2006231}
\bibinfo{author}{Z.~Xu},
\newblock \bibinfo{title}{Induced uncertain linguistic owa operators applied to
  group decision making},
\newblock \bibinfo{journal}{Information Fusion} \bibinfo{volume}{7}
  (\bibinfo{year}{2006}) \bibinfo{pages}{231 -- 238}.
%Type = Article
\bibitem[{Yager(2010)}]{YAGER2010374}
\bibinfo{author}{R.~R. Yager},
\newblock \bibinfo{title}{Lexicographic ordinal owa aggregation of multiple
  criteria},
\newblock \bibinfo{journal}{Information Fusion} \bibinfo{volume}{11}
  (\bibinfo{year}{2010}) \bibinfo{pages}{374 -- 380}.
%Type = Article
\bibitem[{Bustince et~al.(2015)Bustince, Fernandez, Koles\'arov\'a, and
  Mesiar}]{Bustince2015}
\bibinfo{author}{H.~Bustince}, \bibinfo{author}{J.~Fernandez},
  \bibinfo{author}{A.~Koles\'arov\'a}, \bibinfo{author}{R.~Mesiar},
\newblock \bibinfo{title}{Directional monotonicity of fusion functions},
\newblock \bibinfo{journal}{European Journal of Operational Research}
  \bibinfo{volume}{244} (\bibinfo{year}{2015}) \bibinfo{pages}{300 -- 308}.
%Type = Inproceedings
\bibitem[{Hancer et~al.(2015)Hancer, Xue, Zhang, Karaboga, and
  Akay}]{Hancer2015}
\bibinfo{author}{E.~Hancer}, \bibinfo{author}{B.~Xue},
  \bibinfo{author}{M.~Zhang}, \bibinfo{author}{D.~Karaboga},
  \bibinfo{author}{B.~Akay},
\newblock \bibinfo{title}{A multi-objective artificial bee colony approach to
  feature selection using fuzzy mutual information},
\newblock in: \bibinfo{booktitle}{Evolutionary Computation (CEC), 2015 IEEE
  Congress on}, \bibinfo{year}{2015}, pp. \bibinfo{pages}{2420--2427}.
%Type = Inproceedings
\bibitem[{Lingling et~al.(2012)Lingling, Xian, Pengju, and
  Zhigang}]{Lingling2015}
\bibinfo{author}{L.~Lingling}, \bibinfo{author}{Z.~Xian},
  \bibinfo{author}{H.~Pengju}, \bibinfo{author}{L.~Zhigang},
\newblock \bibinfo{title}{The research on the method of fuzzy information
  processing},
\newblock in: \bibinfo{booktitle}{System Science, Engineering Design and
  Manufacturing Informatization (ICSEM), 2012 3rd International Conference on},
  volume~\bibinfo{volume}{2}, \bibinfo{year}{2012}, pp.
  \bibinfo{pages}{47--50}.
%Type = Article
\bibitem[{Bustince et~al.(2013)Bustince, Galar, Bedregal, Koles\'arov\'a, and
  Mesiar}]{Bustince2013}
\bibinfo{author}{H.~Bustince}, \bibinfo{author}{M.~Galar},
  \bibinfo{author}{B.~Bedregal}, \bibinfo{author}{A.~Koles\'arov\'a},
  \bibinfo{author}{R.~Mesiar},
\newblock \bibinfo{title}{A new approach to interval-valued choquet integrals
  and the problem of ordering in interval-valued fuzzy set applications},
\newblock \bibinfo{journal}{IEEE Transactions on Fuzzy Systems}
  \bibinfo{volume}{21} (\bibinfo{year}{2013}) \bibinfo{pages}{1150--1162}.
%Type = Article
\bibitem[{Yager et~al.(2011)Yager, Gumrah, and Reformat}]{Yager2011}
\bibinfo{author}{R.~R. Yager}, \bibinfo{author}{G.~Gumrah},
  \bibinfo{author}{M.~Z. Reformat},
\newblock \bibinfo{title}{Using a web personal evaluation tool - pet for
  lexicographic multi-criteria service selection},
\newblock \bibinfo{journal}{Knowledge-Based Systems} \bibinfo{volume}{24}
  (\bibinfo{year}{2011}) \bibinfo{pages}{929 -- 942}.
%Type = Article
\bibitem[{Beliakov et~al.(2012)Beliakov, Bustince, and
  Paternain}]{Beliakov2012}
\bibinfo{author}{G.~Beliakov}, \bibinfo{author}{H.~Bustince},
  \bibinfo{author}{D.~Paternain},
\newblock \bibinfo{title}{Image reduction using means on discrete product
  lattices},
\newblock \bibinfo{journal}{IEEE Transactions on Image Processing}
  \bibinfo{volume}{21} (\bibinfo{year}{2012}).
%Type = Article
\bibitem[{Joseph et~al.(2014)Joseph, Singh, and Manikandan}]{Joseph2014}
\bibinfo{author}{R.~P. Joseph}, \bibinfo{author}{C.~S. Singh},
  \bibinfo{author}{M.~Manikandan},
\newblock \bibinfo{title}{Brain tumor {MRI} image segmentation and detection in
  image processing},
\newblock \bibinfo{journal}{International Journal of Research and Tecnology}
  \bibinfo{volume}{3} (\bibinfo{year}{2014}).
%Type = Article
\bibitem[{Liang and Xu(2014)}]{Liang2014}
\bibinfo{author}{X.~Liang}, \bibinfo{author}{W.~Xu},
\newblock \bibinfo{title}{Aggregation method for motor drive systems},
\newblock \bibinfo{journal}{Electric Power Systems Research}
  \bibinfo{volume}{117} (\bibinfo{year}{2014}) \bibinfo{pages}{27 -- 35}.
%Type = Book
\bibitem[{Beliakov et~al.(2016)Beliakov, Bunstince, and Calvo}]{Beliakov2016}
\bibinfo{author}{G.~Beliakov}, \bibinfo{author}{H.~Bunstince},
  \bibinfo{author}{T.~Calvo}, \bibinfo{title}{A Practical Guide to Averaging
  Functions}, volume \bibinfo{volume}{329} of \textit{\bibinfo{series}{Studies
  in Fuzziness and Soft Computing}}, \bibinfo{edition}{1} ed.,
  \bibinfo{publisher}{Springer}, \bibinfo{address}{Switzeland},
  \bibinfo{year}{2016}.
%Type = Book
\bibitem[{Dubois and Prade(2000)}]{Dubois2000}
\bibinfo{author}{D.~Dubois}, \bibinfo{author}{H.~Prade},
  \bibinfo{title}{Fundamentals of Fuzzy Sets}, volume~\bibinfo{volume}{7} of
  \textit{\bibinfo{series}{The Handbooks of Fuzzy Sets}}, \bibinfo{edition}{1}
  ed., \bibinfo{publisher}{Springer}, \bibinfo{address}{New York},
  \bibinfo{year}{2000}.
%Type = Book
\bibitem[{Grabisch et~al.(2009)Grabisch, Marichal, Mesiar, and
  Pap}]{Grabisch2009}
\bibinfo{author}{M.~Grabisch}, \bibinfo{author}{J.~L. Marichal},
  \bibinfo{author}{R.~Mesiar}, \bibinfo{author}{E.~Pap},
  \bibinfo{title}{Aggregation Functions}, volume \bibinfo{volume}{127} of
  \textit{\bibinfo{series}{Encyclopedia of Mathematics and its Applications}},
  \bibinfo{publisher}{University Press Cambridge}, \bibinfo{year}{2009}.
%Type = Article
\bibitem[{Mesiar and $\check{S}$pirkov\'a(2006)}]{Messiar2006}
\bibinfo{author}{R.~Mesiar}, \bibinfo{author}{J.~$\check{S}$pirkov\'a},
\newblock \bibinfo{title}{Weighted means and weighting functions},
\newblock \bibinfo{journal}{Kybernetika} \bibinfo{volume}{42}
  (\bibinfo{year}{2006}) \bibinfo{pages}{151--160}.
%Type = Article
\bibitem[{Mesiar et~al.(2008)Mesiar, $\check{S}$pirkov\'a, and
  Vavr\'ikov\'a}]{Mesiar2008}
\bibinfo{author}{R.~Mesiar}, \bibinfo{author}{J.~$\check{S}$pirkov\'a},
  \bibinfo{author}{L.~Vavr\'ikov\'a},
\newblock \bibinfo{title}{Weighted aggregation operators based on
  minimization},
\newblock \bibinfo{journal}{Information Sciences} \bibinfo{volume}{178}
  (\bibinfo{year}{2008}) \bibinfo{pages}{1133 -- 1140}.
%Type = Inproceedings
\bibitem[{Wilkin and Beliakov(2014)}]{Wilkin2014}
\bibinfo{author}{T.~Wilkin}, \bibinfo{author}{G.~Beliakov},
\newblock \bibinfo{title}{Weakly monotone averaging functions},
\newblock in: \bibinfo{booktitle}{IPMU'2014}, volume \bibinfo{volume}{444},
  \bibinfo{year}{2014}, pp. \bibinfo{pages}{364--373}.
%Type = Article
\bibitem[{Wilkin and Beliakov(2015)}]{Wilkin2015}
\bibinfo{author}{T.~Wilkin}, \bibinfo{author}{G.~Beliakov},
\newblock \bibinfo{title}{Weakly monotone aggregation functions},
\newblock \bibinfo{journal}{International Journal of Intelligent. Systems}
  \bibinfo{volume}{30} (\bibinfo{year}{2015}) \bibinfo{pages}{144--169}.
%Type = Article
\bibitem[{Lucca et~al.(2016)Lucca, Sanz, Dimuro, Bedregal, Mesiar,
  Koles\'arov\'a, and Bustince}]{Luca2016}
\bibinfo{author}{G.~Lucca}, \bibinfo{author}{J.~A. Sanz},
  \bibinfo{author}{G.~P. Dimuro}, \bibinfo{author}{B.~Bedregal},
  \bibinfo{author}{R.~Mesiar}, \bibinfo{author}{A.~Koles\'arov\'a},
  \bibinfo{author}{H.~Bustince},
\newblock \bibinfo{title}{Preaggregation functions: Construction and an
  application},
\newblock \bibinfo{journal}{IEEE Transactions on Fuzzy Systems}
  \bibinfo{volume}{24} (\bibinfo{year}{2016}) \bibinfo{pages}{260--272}.
%Type = Book
\bibitem[{Gonzales and Woods(2008)}]{Gonzales}
\bibinfo{author}{R.~C. Gonzales}, \bibinfo{author}{R.~E. Woods},
  \bibinfo{title}{Digital Image Processing}, \bibinfo{edition}{3rd} ed.,
  \bibinfo{publisher}{Pearson}, \bibinfo{address}{New Jersey},
  \bibinfo{year}{2008}.
%Type = Article
\bibitem[{Keys(1981)}]{Keys1981}
\bibinfo{author}{R.~Keys},
\newblock \bibinfo{title}{Cubic convolution interpolation for digital image
  processing},
\newblock \bibinfo{journal}{IEEE Transactions on Acoustics, Speech, and Signal
  Processing} \bibinfo{volume}{29} (\bibinfo{year}{1981})
  \bibinfo{pages}{1153--1160}.
%Type = Article
\bibitem[{Lehmann et~al.(1999)Lehmann, Gonner, and Spitzer}]{Lehmann1999}
\bibinfo{author}{T.~M. Lehmann}, \bibinfo{author}{C.~Gonner},
  \bibinfo{author}{K.~Spitzer},
\newblock \bibinfo{title}{Survey: interpolation methods in medical image
  processing},
\newblock \bibinfo{journal}{IEEE Transactions on Medical Imaging}
  \bibinfo{volume}{18} (\bibinfo{year}{1999}) \bibinfo{pages}{1049--1075}.
%Type = Article
\bibitem[{Thevenaz et~al.(2000)Thevenaz, Blu, and Unser}]{Thevenaz2000}
\bibinfo{author}{P.~Thevenaz}, \bibinfo{author}{T.~Blu},
  \bibinfo{author}{M.~Unser},
\newblock \bibinfo{title}{Interpolation revisited},
\newblock \bibinfo{journal}{IEEE Transactions on Medical Imaging}
  \bibinfo{volume}{19} (\bibinfo{year}{2000}) \bibinfo{pages}{739--758}.

\end{thebibliography}

\section{Appendix: Illustrative example of application (Image reduction)}

In this part of our work we use the \textsf{GM} functions $Min$, $Max$, $Arith$, $Med$, $\textsf{cOWA}$ and \textbf{H} to build image reduction operators. A broad discussion on this field of application can be found in \cite{Gonzales}. Here, we are just interested to shows a possibility of application for our functions.

\subsection{Methodological process}

The methodological process is the same as that used in \cite{Paternain2015}, it consists of:

\begin{enumerate}
	\item Reduce the input images using the $Min$, $Max$, $Arith$, $Med$, \textsf{cOWA} and $\mathbf{H}$;
	\item Magnify the reduced images to the original size using three different method: (1) The nearest neighbor interpolation, (2) The bilinear interpolation and (2) the bicubic interpolation (see \cite{Gonzales,Keys1981,Lehmann1999,Thevenaz2000}) ;
	\item Compare the output images with the original one using the measure $PSNR$.
\end{enumerate}

We use ten original images, in grayscale, of size $512\times 512$. The obtained results are shown below (The bold value represents the high quality image, and the italic value represents the second high quality image):

\subsection{Results}

\begin{table}[!htb]
	\centering
	\tiny
	\begin{center}
		\hspace{2.5cm} USING $2\times 2$ BLOCKS \hspace{4.8cm} USING $4\times 4$ BLOCKS
	\end{center}
	\begin{tabular}{ccccccc|cccccc}
		& $Min$ & $Max$ & $Med$ & $Arith$ & \textsf{cOWA} & $\mathbf{H}$ & $Min$ & $Max$ & $Med$ & $Arith$ & \textsf{cOWA} & $\mathbf{H}$\\
		\hline\
		Img 01 & 26,68848 & 26,60371 & 30,66996 & {\bf 30,89667} & 30,73823 & {\it 30,75448} & 21,37117 &	20,83960 &	26,73708 &	{\bf 27,07854} &	27,01270 &	{\it 27,07067}\\
		Img 02 & 33,50403 &	33,46846 &	37,51525 &	{\bf 37,64240} &	37,57713 &	{\it 37,58138} & 19,70858 &	19,54290 &	23,92198 &	{\bf 24,07786} &	24,05762 &	{\it 24,07478}\\
		Img 03 & 26,80034 &	26,74460 &	30,47904 &	{\bf 30,55504} & {\it 30,52128} &	30,51564 & 20,46198 &	20,82576 &	25,64113 &	{\bf 26,16092} &	26,08186 &	{\it 26,14607}\\
		Img 04 & 28,90415 &	28,83284 &	32,88120 &	{\bf 33,01225} &	{\it 32,94828} &	32,94146 & 22,59335 &	22,24354 &	27,94347 &	{\bf 28,26449} &	28,19574 &	{\it 28,25700}\\
		Img 05 & 25,04896 &	25,04438 &	28,75582 &	{\bf 28,85475} &	{\it 28,81506} &	28,79901 & 18,86628 &	19,55278 &	24,12507 &	{\bf 24,68962} &	24,58713 &	{\it 24,67322}\\
		Img 06 & 38,10156 &	38,07248 &	42,08612 &	{\bf 42,13003} &	{\it 42,12316} &	42,11653 & 29,48308 &	29,26559 &	34,89670 &	{\bf 35,11481} &	35,09436 &	{\it 35,11023}\\
		Img 07 & 24,48520 &	24,38872 &	28,31229 &	{\bf 28,45667} &	28,35114 &	{\it 28,37668} & 18,95771 &	18,72670 &	24,18918 &	{\bf 24,55073} &	24,48373 &	{\it 24,54269}\\
		Img 08 & 23,69576 &	23,73464 &	27,41557 &	{\bf 27,51579} &	{\it 27,46383} &	27,45864 & 17,71071 &	18,59348 &	23,11305 &	{\bf 23,54332} &	23,43522 &	{\it 23,53119}\\
		Img 09 & 26,19262 &	26,09448 &	30,06427 &	{\bf 30,22940} &	30,11893 &	{\it 30,13332} & 20,97846 &	20,44416 &	26,23824 &	{\bf 26,53197} &	26,42064 &	{\it 26,52562}\\
		Img 10 & 21,48459 &	21,41350 &	25,37475 &	{\bf 25,58054} &	25,43016 &	{\it 25,45073} & 16,47636 &	16,22205 &	21,89755 &	{\bf 22,22614} &	22,10356 &	{\it 22,21825}\\
		\hline\
		Avg & 27,49057 &	27,43978 &	31,35543 &	{\bf 31,48735} &	31,40872 &	{\it 31,41279} & 20,66077 &	20,62565 &	25,87034 &	{\bf 26,22384} &	26,14726 & {\it 26,21497}
	\end{tabular}
	\caption{$PSNR$ values of reconstruction of imagens by nn interpolation.}
	
	\begin{center}
		\hspace{2.5cm} USING $2\times 2$ BLOCKS \hspace{4.8cm} USING $4\times 4$ BLOCKS
	\end{center}
	\begin{tabular}{ccccccc|cccccc}
		& $Min$ & $Max$ & $Med$ & $Arith$ & \textsf{cOWA} & $\mathbf{H}$ & $Min$ & $Max$ & $Med$ & $Arith$ & \textsf{cOWA} & $\mathbf{H}$\\
		\hline\
		Img 01 & 27,25658	& 27,41249 & {\it 31,70137} & 31,66148 & 31,64818 & {\bf 31,70944} & 21,84394 & 21,46624 & {\it 28,12885} & 28,03911 & {\bf 28,13262} & 28,08806\\
		Img 02 & 29,07393 &	29,09065 &	29,98667 &	{\bf 30,00618} &	{\it 29,99790} &	29,99295 & 20,22210 &	19,99324 &	24,09349 &	24,09114 &	{\it 24,09696} &	{\bf 24,10058}\\
		Img 03 & 28,07377 &	27,53953 &	{\bf 31,96271} &	31,87901 &	31,87085 &	{\it 31,94673} & 21,36383 &	21,65788 &	27,34577 &	27,53279 &	{\bf 27,57114} &	{\it 27,56163}\\
		Img 04 & 29,70934 &	29,78913 &	{\bf 34,39128} &	34,28215 &	34,31414 &	{\it 34,37504} & 23,23057 &	22,96007 &	{\bf 29,81717} &	29,65596 &	{\it 29,77096} &	29,71475\\
		Img 05 & 26,30684 &	25,74955 &	{\bf 30,17965} &	30,08193 &	30,05530 &	{\it 30,16533} & 19,54307 &	20,06159 &	25,32192 &	25,47922 &	{\it 25,51400} &	{\bf 25,51442}\\
		Img 06 & 40,09734 &	39,94107 &	{\bf 48,99047} &	48,55730 &	48,52986 &	{\it 48,86710}  & 30,92215 &	30,60188 &	{\bf 42,72668} &	41,77064 &	{\it 41,99358} &	41,97442\\
		Img 07 & 25,10689 &	25,04408 &	{\it 28,93328} &	28,92340 &	28,89276 &	{\bf 28,94254} & 19,43662 &	19,19604 &	24,96897 &	25,00413 &	{\bf 25,05911} &	{\it 25,02899}\\
		Img 08 & 24,63619 &	24,10410 &	{\it 28,19100} &	28,17758 &	28,16818 &	{\bf 28,19312} & 18,28578 &	18,86696 &	23,87169 &	{\it 24,09781} &	24,07356 &	{\bf 24,10310}\\
		Img 09 & 26,60297 &	26,71398 &	30,54028 &	{\bf 30,56126} &	30,52693 &	{\it 30,55733} & 21,32747 &	20,91360 &	27,09762 &	27,10526 &	{\bf 27,16280} &	{\it 27,13073}\\
		Img 10 & 21,93973 &	21,90280 &	25,71329 &	{\bf 25,74295} &	25,69402 &	{\it 25,73353} & 16,77848 &	16,57833 &	22,58040 &	22,61488 &	{\it 22,63949} &	{\bf 22,63987}\\
		\hline\
		Avg & 27,88036 &	27,72874 &	{\bf 32,05900} &	31,98732 &	31,96981 &	{\it 32,04831} & 21,29540 &	21,22958 &	{\it 27,59525} &	27,53909 &	{\bf 27,60142} &	27,58566
	\end{tabular}
	\caption{$PSNR$ values of reconstruction of imagens by bilinear interpolation.}
	
	\begin{center}
		\hspace{2.5cm} USING $2\times 2$ BLOCKS \hspace{4.8cm} USING $4\times 4$ BLOCKS
	\end{center}
	\begin{tabular}{ccccccc|cccccc}
		
		& $Min$ & $Max$ & $Med$ & $Arith$ & \textsf{cOWA} & $\mathbf{H}$ & $Min$ & $Max$ & $Med$ & $Arith$ & \textsf{cOWA} & $\mathbf{H}$\\
		\hline\
		Img 01 & 27,39667 & 27,45993 & 32,53367 & {\bf 32,62657} & 32,52946 & {\it 32,58602} & 21,83423 & 21,39364 & 28,64265 & 28,74908 & {\bf 28,80893} & {\it 28,78768}\\
		Img 02 & 30,06149 & 30,00816 & 31,28820 & {\bf 31,31873} & {\it 31,30611} & 31,29877  & 20,20038 & 19,88701 & 24,49596 & {\it 24,56989} & 24,56761 & {\bf 24,57359}\\
		Img 03 & 28,09952 & 27,62931 & {\it 32,92967} & 32,90897 & 32,87767 & {\bf 32,93859} & 21,25132 & 21,55589 & 27,82091 & {\it 28,31402} & 28,28961 & {\bf 28,32229}\\
		Img 04 & 29,92114 & 29,94430 & {\it 35,70586} & 35,70361 & 35,68906 & {\bf 35,73313} & 23,22310 & 22,89860 & 30,47704 & 30,54773 & {\bf 30,60332} & {\it 30,59348}\\
		Img 05 & 26,38597 & 25,93655 & {\it 31,32017} & 31,30790 & 31,25508 & {\bf 31,33640} & 19,45423 & 20,06391 & 25,74518 & {\it 26,18606} & 26,15139 & {\bf 26,20092}\\
		Img 06 & 40,05229 & 40,02173 & {\bf 51,35284} & 51,07478 & 51,01447 & {\it 51,31081} & 30,81953 & 30,48357 & {\bf 44,31891} & 43,83439 & 44,03526 & {\it 44,05492}\\
		Img 07 & 25,23188 & 25,16984 & 29,85564 & {\bf 29,93609} & 29,85733 & {\it 29,89915} & 19,36949 & 19,11221 & 25,29211 & 25,49221 & {\it 25,49999} & {\bf 25,50641}\\
		Img 08 & 24,72669 & 24,32047 & 29,10402 & {\bf 29,15066} & 29,11737 & {\it 29,12822} & 18,21007 & 18,91559 & 24,17857 & {\bf 24,57330} & 24,49174 & {\it 24,56575}\\
		Img 09 & 26,73252 & 26,79140 & 31,27454 & {\bf 31,38274} & 31,29368 & {\it 31,32452} & 21,32252 & 20,85345 & 27,41366 & {\it 27,56839} & 27,55860 & {\bf 27,58354}\\
		Img 10 & 22,04218 & 21,98136 & 26,39147 & {\bf 26,52171} & 26,41585 & {\it 26,44659} & 16,76501 & 16,53815 & 22,82004 & {\it 23,00025} & 22,96201 & {\bf 23,01459}\\
		\hline\
		Avg & 28,06504 & 27,92630 & 33,17561 & {\it 33,19318} & 33,13561 & {\bf 33,20022} & 21,24499 & 21,17020 & 28,12050 & 28,28353 & {\it 28,29685} & {\bf 28,32032}
	\end{tabular}
	\caption{$PSNR$ values of reconstruction of imagens by bicubic interpolation.}
\end{table}

\end{document}